\documentclass[a4paper,12pt]{scrartcl}

\usepackage{actuarialsymbol}
\usepackage{amsmath}
\usepackage{amssymb}
\usepackage{amsthm}
\usepackage{bbm}
\usepackage[
    backend=biber, 
    articlein=false, 
    dashed=false, 
    giveninits=true,
    maxcitenames=2, 
    natbib=true, 
    style=ext-authoryear-ibid, 
    uniquename=init,
    ]{biblatex}
\usepackage{caption}
\usepackage{chngcntr}
\usepackage{enumitem}
\usepackage[margin=2.5cm]{geometry}
\usepackage{graphicx}
\usepackage{subdepth}
\usepackage[table]{xcolor}

\allowdisplaybreaks
\counterwithin{figure}{section}
\citetrackerfalse 
\newtheoremstyle{NoItalic}{}{}{}{}{\bfseries}{.}{.5em}{\thmname{#1}\thmnumber{ #2}\thmnote{ #3}}\theoremstyle{NoItalic}

\newtheorem{pro}{Proposition}[section]

\makeatletter 
    \let\c@pro=\c@figure 
    \xpatchcmd{\proof}{\@addpunct{.}}{\@addpunct{:}}{}{} 
\makeatother

\DefineBibliographyStrings{english}{and={\&}} 
\DeclareFieldFormat*{title}{#1} 
\DeclareFieldFormat[inbook]{booktitle}{#1} 
\DeclareFieldFormat[article,periodical]{number}{\mkbibparens{#1}} 
\DeclareNameAlias{sortname}{family-given} 

\begin{document} 

\title{A note on bequest preferences in utility maximisation for modern tontines}
\author{Thomas Bernhardt\footnote{University of Manchester, {\texttt{thomas.bernhardt@manchester.ac.uk}}}} 
\date{}
\maketitle
\vspace{-4em}

\section*{\begin{center}Abstract\end{center}}
\begin{abstract}
\noindent In this short note, we address two issues in the literature about modern tontines with bequest and utility maximisation: how to verify optimal controls and the decreasing allocation of funds in the tontine. We want to raise awareness in the actuarial community about the dual approach to solve optimal control problems when working with power utilities. Additionally, we point out that bequest preferences should be time-dependent or otherwise yield unrealistic investment strategies. We base our attempt at modelling bequest preferences on common sense rules like 100\% payback upon death at the start that vanishes over time. Our modelling shows that the resulting investment strategy almost linearly adjusts the allocation in the tontine from 0\% to 100\% over time.
\\[1em]
\textbf{Keywords:} Tontine; bequest; utility maximisation; dual approach; linear adjustment.
\\[0.5em]
\textbf{JEL codes:} C61, G23
\end{abstract}

\section{Introduction}\label{section:Intro}

Tontines are cost-efficient and sustainable alternatives to guaranteed pensions and annuities. They have no risk margin and do not rely on the accuracy of initial best estimates in theory. They mitigate any risk of providing a lifelong income by adjusting the income in line with experienced outcomes. They transfer systematic risks directly to the income while mitigating unsystematic risks by redistributing funds. Tontines are promising candidates to organise retirement in an ageing society with DC pension pots.

The lack of guarantees allowing tontines to adjust members' income gives them freedom in their design. For example, tontines could invest freely in the financial market in contrast to pensions and annuities, which are heavily regulated. The design flexibility extends to issues unrelated to tontines but are considered issues with pensions and annuities like bequests.

In this paper, we model bequest preferences for a tontine that invests in a financial market. We point out why the previous modelling appears inadequate and solve a previously observed technical problem that helps with the preferences.

Papers modelling tontines with bequest include \citet{BernhardtDonnelly2019}, \citet{Dagpunar2021} and \citet{ChenRach2022}. Even though the settings differ, all authors report a low fund allocation towards bequest for reasonable parameters. \citet{BernhardtDonnelly2019} pay a fixed proportion of remaining funds as bequest; that proportion ranges from 0 to 20\%. \citet{Dagpunar2021} pays a variable proportion of remaining funds as bequest; that ranges initially from 0 to 20\%. \citet{ChenRach2022} split the funds between a tontine and a life insurance product; 5.5\% go to life insurance. Additionally, if investors can directly control allocation or payout, they would increase the allocation towards or payout from bequest over time for reasonable parameters. To be clear, \citet{ChenRach2022} mention that the bequest payout can be constant, decreasing or increasing over time depending on the interplay between interest rate and subjective discount rate. However, the subjective discount rate is not a free parameter when we model utility from discounted wealth, which sets the subjective discount rate so that the bequest payout will increase in that setting. One can say that the literature suggests that investors initially allocate a small fraction to bequeath but would increase that fraction or payout over time.

The findings from the literature appear counterintuitive because we would expect that modelling a bequest motive leads to a significant allocation of funds to it. \citet{DavidoffBrownDiamond2005} already suggested that something is wrong with our modelling in the context of the annuity puzzle. Here, only a small proportion of funds is set aside as bequest while the maturity of funds is annuitised unless we assume extreme situations similar to before. The underlying problem is dying at an early age is unlikelier than dying at an advanced age. Thus, maximising income and bequest favours bequest at an advanced age when dying is likelier than at an early age when it is unlikely.

A bequest motive should model the fear of losing savings upon an unexpected death rather than maximising its output. We think of countering unexpected events rather than capitalising on the most likely events. For example, it is unacceptable to invest in a pension and lose all savings in the next moment because of sudden death. But it is acceptable to live off savings and not bequeath when living longer than expected. In particular, bequest preferences should be strong at the start and weaken over time in retirement, yielding a large proportion of remaining funds in the beginning and a decreasing proportion over time, i.e.\ the exact opposite of the findings in the recent literature. 

It is worth mentioning the Riccati tontine by \citet{MilevskySalisbury2024}. Members who pass away during the investment period receive their initial investment (but not the shared investment returns) on average back. The proportion (called surrender charge schedule) from total investments paid to the deceased has the properties we want, decreasing over time with a payback guarantee at the start. However, the properties are an artefact of a fixed average payout and above-inflation returns from the financial market. For example, there would be no decrease or sharing without the financial market.

We want to find preferences that yield large bequests at the start that decrease over time. Previous research modelled the balance between income and bequest to heirs with a fixed parameter. If we want a strong bequest motive at the start that weakens over time, we need time-dependent preferences. However, establishing the optimal allocation and investment strategies is a problem in \citet[p.264]{Dagpunar2021} with fixed preferences.

Candidate controls follow from an HJB approach, but a verification step is needed to show their optimality. A PDE that describes the problem locally implies the candidate controls. Because of the random nature of investments, the definition of global optimality uses an expectation. The problem is that connecting the PDE to the objective yields a limit of expectations with potentially irregular asymptotic. It is well-known that Fatou's lemma can deal with that limit when maximising a non-negative objective. However, Fatou's lemma does not apply to maximising objectives with negative numbers like power utility with a negative exponent, i.e.\ the case of risk-averse investors.

The case of risk-averse investors was a thorny issue even in Merton's portfolio problem until \citet{Pliska1986} with his dual approach. In a nutshell, Fatou's lemma does not apply to the original problem but a related one. However, it relies on a concave utility function. \citet{Biagini2010} gives an abstract overview of the approach without a concrete example. \citet{Rogers2013} discusses many generalised versions of the original Merton problem where the method solved the problem; see Chapter 1.2 in particular. See \citet{HerdegenHobsonJerome2021} for a recent take on the problem and alternative approaches. It is worth noting most authors omit Fatou's lemma and choose to state that a positive local martingale is a supermartingale. We use the dual approach to verify the solution by \citet[p.264]{Dagpunar2021} and generalise it to the case of time-dependent preferences.

We use the setting of \citet{Dagpunar2021} but with a time-dependent deterministic function that models the relative importance between income and potential bequest payout in the rest of the paper. More precisely, we look at an idealised tontine with infinitely many members and known mortality distribution, i.e.\ there are no discretisation issues with the redistribution of funds. The investors choose the proportion of funds used for redistribution upon death and, at the same time, determine how much they receive from the tontine in return when alive. The tontine pays the remaining funds blocked from redistribution as a bequest upon the investors' death. It is worth noting that the investors' bequest grows as income does from redistribution. The investors control their investment in a Black-Scholes market and income from their remaining funds. For mathematical tractability, we assume that investors seek to maximise the lifetime utility of their consumption and bequest using constant relative risk aversion.  

In section 2, after introducing our setting mathematically, we state the candidate controls of our optimisation problem and verify the optimality with the dual approach.

In section 3, we choose parameters and perform a modelling exercise based on common sense ideas related to bequest in retirement. For example, instead of exploiting the likelihood of events, we avoid situations where an investor loses a large proportion of their funds before receiving an income for many years. Our modelling exercise shows that having a full refund at the start and a declining bequest reaching zero in the following years leads to an almost linear allocation of funds in the tontine.  

We finish the paper with a summary of our results in section 4.

\section{Mathematical description and explicit solution}\label{section:setting}

We consider an investor with initial wealth $\,X_0>0\,$ and a deterministic force of mortality $\,\lambda_t\geq0\,$ for $\,t\geq0$. In particular, the survival probability is given by $\,S_t=\exp(-\int_0^t\lambda_u\,\mathrm{d}u)$.

We consider a financial market with a stock and a bank account. The bank account has a constant continuous interest rate $\,r\in\mathbb{R}\,$ and the stock price follows a geometric Brownian motion with constant drift $\,\mu>r\,$ and constant volatility $\,\sigma>0$.

The investor can dynamically decide on the proportion of wealth $\pi_t$ to invest in the stock, the proportion to consume $\,c_t\geq0$, and the proportion in the tontine $\,\alpha_t\leq1$. The proportion in the tontine yields an additional stream of funds to the investor's wealth through the redistribution of wealth from recently deceased members in the tontine. All controls might be random as long as they depend on the available information.

We neglect any transaction costs and assume no borrowing or short-selling constraints. We neglect any discretisation problems by assuming an infinitely large pool, i.e.\ the local rate of the number of recently deceased is $\lambda_t$. In particular, the dynamics of the investor's wealth process $X_t$ is
\begin{equation}\label{eq:X}
\begin{gathered}
    \mathrm{d}X_t/X_t=(r+(\mu-r)\pi_t-c_t+\alpha_t\lambda_t)\mathrm{d}t+\sigma\pi_t\mathrm{d}W_t,
    \\X_t=X_0\exp\Big(\int_0^t\!r+(\mu-r)\pi_u-c_u+\alpha_u\lambda_u\,\mathrm{d}u\Big)\,\mathcal{E}_t\Big(\int_0\sigma\pi_u\,\mathrm{d}W_u\Big),
\end{gathered}    
\end{equation}
in which $W_t$ is a Brownian motion and $\mathcal{E}$ stands for the stochastic exponential. 

To make mathematical sense out of (\ref{eq:X}), we assume that there is an underlying probability space $(\Omega,\mathcal{F},\mathbb{P})$ with filtration $(\mathcal{F}_t)_{t\geq0}$ for the Brownian motion and the controls. In particular, we assume that the controls $\,\pi_t,c_t,\alpha_t$ are admissible, i.e.\ adapted to $(\mathcal{F}_t)_{t\geq0}$ and $\,\int_0^t\pi_u^2\,\mathrm{d}u$, $\int_0^t|c_u|\,\mathrm{d}u$, $\int_0^t|\alpha_u|\,\mathrm{d}u<\infty\,$ for all $\,t\geq0\,$ and outcomes from $\Omega$.

The investor looks at the average sum of lifetime utility from consumption and bequest to decide on the controls $\,\pi_t,c_t,\alpha_t$. More precisely, we use power utility and try to maximise
\begin{equation*}
    \mathbb{E}\big[\int_0^\tau\!\mathrm{e}^{-\rho u}\big(c_uX_u\big)^{\gamma}/\gamma\,\mathrm{d}u+\mathrm{e}^{-\rho\tau}b_\tau\big((1-\alpha_\tau)X_\tau\big)^{\gamma}/\gamma\big],
\end{equation*}
in which $\tau$ is the random remaining lifetime of the investor, $\rho$ is the constant subjective discount rate, $\,\gamma\in(-\infty,1)\setminus{0}\,$ is the investor's constant relative risk aversion, and $\,b_t\geq0\,$ is a deterministic function that quantifies the investor's preference towards consumption or bequest. Note that (\ref{eq:X}) and $\,c_t\geq0\,$ and $\,\alpha_t\leq1\,$ imply that the income rate $\,c_tX_t\geq0\,$ and the bequest $\,(1-\alpha_t)X_t\geq0$, i.e.\ $(c_uX_u)^{\gamma}$ and $\,((1-\alpha_\tau)X_\tau)^{\gamma}$ are well-defined with the convention that $\,0^\gamma\,$ is either 0 or $\infty$ depending on the sign of $\gamma$.

We assume that the remaining lifetime $\tau$ is independent of the financial market so that we can assume that $\tau$ is independent of $(\mathcal{F}_t)_{t\geq0}$, see \citet[Lemma A.1]{BernhardtDonnelly2019} for more details. Now, we can remove $\tau$ from the average sum of consumption and bequest, i.e.\ our objective function is
\begin{equation}\label{eq:objective}
    \mathbb{E}\big[\int_0^\infty\!\mathrm{e}^{-\rho u}S_u\big(c_uX_u\big)^{\gamma}/\gamma+\mathrm{e}^{-\rho u}S_u\lambda_ub_u\big((1-\alpha_u)X_u\big)^{\gamma}/\gamma\,\mathrm{d}u\big].
\end{equation}

\begin{pro}\label{pro:controls}
Let $\,\beta=r+\frac{\rho-r}{1-\gamma}-\frac{1}{2}\frac{\gamma}{(1-\gamma)^2}(\frac{\mu-r}{\sigma})^2$.
Assume
\[\int_0^\infty\mathrm{e}^{-\beta u}S_u(1+b_u^{\frac{1}{1-\gamma}}\lambda_u)\,\mathrm{d}u<\infty.\]
Then the optimal controls of the objective (\ref{eq:objective}) are
\begin{align*}
    \pi^*_t&=\frac{1}{1-\gamma}\frac{\mu-r}{\sigma^2},
    \\c^*_t&=\frac{\mathrm{e}^{-\beta t}S_t}{\int_t^\infty\mathrm{e}^{-\beta u}S_u(1+b_u^{\frac{1}{1-\gamma}}\lambda_u)\,\mathrm{d}u},
    \\1-\alpha^*_t&=\frac{\mathrm{e}^{-\beta t}S_t\,b_t^{\frac{1}{1-\gamma}}}{\int_t^\infty\mathrm{e}^{-\beta u}S_u(1+b_u^{\frac{1}{1-\gamma}}\lambda_u)\,\mathrm{d}u}.
\end{align*}
\end{pro}
\begin{proof}
    We compare an arbitrary state process $X$ corresponding to controls $\pi,c,\alpha$ with the process $X^*$ corresponding to the explicit candidate controls $\,\pi^*,c^*,\alpha^*$ from Proposition \ref{pro:controls}. We use the concavity of the utility function, more precisely, since $\,0\neq\gamma<1$,
    \begin{equation}\label{eq:concavity}
        x^\gamma/\gamma\leq y^\gamma/\gamma + y^{\gamma-1}(x-y)\quad\text{for all $\,x\geq0\,$ and $\,y>0$}.
    \end{equation}
    Let us begin by estimating the objective function of $X$ with the objective function of $X^*$,
    \begin{equation}\label{eq:estimate}
    \begin{aligned}
        \mathbb{E}&\big[\int_0^\infty\!\mathrm{e}^{-\rho u}S_u\big(c_uX_u\big)^{\gamma}/\gamma+\mathrm{e}^{-\rho u}S_u\lambda_ub_u\big((1-\alpha_u)X_u\big)^{\gamma}/\gamma\,\mathrm{d}u\big]
        \\\leq&\;\mathbb{E}\big[\int_0^\infty\!\mathrm{e}^{-\rho u}S_u\big(c^*_uX^*_u\big)^{\gamma}/\gamma+\mathrm{e}^{-\rho u}S_u\lambda_ub_u\big((1-\alpha^*_u)X^*_u\big)^{\gamma}/\gamma\,\mathrm{d}u\big]
        \\&+\mathbb{E}\big[\int_0^\infty\!\mathrm{e}^{-\rho u}S_u\big(c^*_uX^*_u\big)^{\gamma-1}(c_uX_u-c_u^*X_u^*)\,\mathrm{d}u\big]
        \\&+\mathbb{E}\big[\int_0^\infty\!\mathrm{e}^{-\rho u}S_u\lambda_ub_u\big((1-\alpha^*_u)X^*_u\big)^{\gamma-1}((1-\alpha_u)X_u-(1-\alpha^*_u)X_u^*)\,\mathrm{d}u\big].
    \end{aligned}
    \end{equation}
    Note that there is an integrability issue in the above inequality and that the above inequality only holds if the sum of the last two summands is bounded. Also, note that the times when $\,1-\alpha_t^*=0\,$ is not an issue with equation (\ref{eq:concavity}) because $\,1-\alpha^*_t=0\,$ if and only if $\,b_t=0$. 
    
    It is enough to show that the sum of the last two summands in (\ref{eq:estimate}) is negative and bounded to conclude that (\ref{eq:estimate}) holds and that the objective function of any controls is smaller than the objective function of the candidate controls. Let 
    \begin{equation}\label{eq:varphi}
        \varphi_t=S_t/(c^*_t)^{1-\gamma},
    \end{equation}
    then 
    \begin{equation}\label{eq:c^*=varphi}
    \begin{aligned}
        c^*_t&=(S_t/\varphi_t)^\frac{1}{1-\gamma},
        \\1-\alpha^*_t&=(S_t/\varphi_t)^\frac{1}{1-\gamma}(b_t)^\frac{1}{1-\gamma},
    \end{aligned}
    \end{equation}
    thus the last two summands in (\ref{eq:estimate}) simplify to
    \[
        \mathbb{E}\big[\int_0^\infty\!\mathrm{e}^{-\rho u}\varphi_u(X^*_u)^{\gamma-1}\big(c_uX_u+\mathbbm{1}_{b_u>0}\lambda_u(1-\alpha_u)X_u-c_u^*X_u^*-\mathbbm{1}_{b_u>0}\lambda_u(1-\alpha^*_u)X_u\big)\,\mathrm{d}u\big].
    \]
    Moreover, we can assume that $\,1-\alpha_t=0\,$ whenever $\,b_t=0\,$ without loss of generality because those two terms appear in the original objective function (\ref{eq:estimate}) together. Thus, the last two summands in (\ref{eq:estimate}) can be expressed as
    \begin{equation}\label{eq:noindicators}
        \mathbb{E}\big[\int_0^\infty\!\mathrm{e}^{-\rho u}\varphi_u(X^*_u)^{\gamma-1}\big(c_uX_u+\lambda_u(1-\alpha_u)X_u-c_u^*X_u^*-\lambda_u(1-\alpha^*_u)X_u\big)\,\mathrm{d}u\big].
    \end{equation}
    
    We investigate the factor $\,\zeta_t=\mathrm{e}^{-\rho t}\varphi_t(X^*_t)^{\gamma-1}$ in the above integrand. It is called the state-price density process in the optimal control literature. It is linked to the value function through its first partial derivative with respect to the state. First, we derive an ODE for $\varphi$. From (\ref{eq:varphi}) and (\ref{eq:c^*=varphi}), we get    
    \begin{equation}\label{eq:varphi'}
    \begin{gathered}
        \varphi_t'/\varphi_t=(1-\gamma)(\beta-c^*_t+\alpha^*_t\lambda_t)-\lambda_t,
        \\\varphi_t=\varphi_0\exp\Big((1-\gamma)\int_0^t\!\beta-c^*_u+\alpha^*_u\lambda_u\,\mathrm{d}u-\int_0^t\!\lambda_u\,\mathrm{d}u\Big).
    \end{gathered}
    \end{equation}
    Combining (\ref{eq:X}), (\ref{eq:varphi'}) with $\beta,\,\pi^*$ from Proposition \ref{pro:controls} yields
    \begin{equation}\label{eq:zeta}
    \begin{gathered}
        \zeta_t=\mathrm{e}^{-\rho t}\varphi_t(X^*_t)^{\gamma-1}=\varphi_0\mathrm{e}^{-rt}\exp\Big(-\int_0^t\!\lambda_u\,\mathrm{d}u\Big)\,\mathcal{E}_t\Big(-\frac{\mu-r}{\sigma}W\Big),
        \\\mathrm{d}\zeta_t/\zeta_t=-(r+\lambda_t)\mathrm{d}t-\frac{\mu-r}{\sigma}\mathrm{d}W_t.
    \end{gathered}
    \end{equation}
    Now, we are ready to showcase the main argument of the verification. We define and derive the SDE of the following process using (\ref{eq:X}) and (\ref{eq:zeta}),
    \begin{equation}\label{eq:Y}
    \begin{gathered}
        Y_t=\zeta_tX_t+\int_0^t\!\zeta_u(c_uX_u+\lambda_u(1-\alpha_u)X_u)\,\mathrm{d}u,
        \\
        \begin{aligned}
            \mathrm{d}Y_t=&\,\zeta_t\mathrm{d}X_t+X_t\mathrm{d}\zeta_t+\mathrm{d}[\zeta,X]_t+\zeta_t(c_tX_t+\lambda_t(1-\alpha_t)X_t)\mathrm{d}t
            \\=&\,(r+(\mu-r)\pi_t-c_t+\alpha_t\lambda_t)\zeta_tX_t\mathrm{d}t+\text{local martingale} 
            \\&-(r+\lambda_t)\zeta_tX_t\mathrm{d}t+\text{local martingale} 
            \\&-(\mu-r)\pi\zeta_tX_t\mathrm{d}t
            \\&+\zeta_t(c_tX_t+\lambda_t(1-\alpha_t)X_t)\mathrm{d}t
            \\=&\,\text{local martingale}.
        \end{aligned}
    \end{gathered}
    \end{equation}
    Note that $\,\zeta_tX_t$ as well as $\,\zeta_t(c_tX_t+\lambda_t(1-\alpha_t)X_t)\,$ are positive. Thus, $Y$ is a supermartingale because it is a positive local martingale, and furthermore
    \begin{equation}\label{eq:Y0}
    \begin{aligned}
        Y_0&\geq\liminf_{t\rightarrow\infty}\mathbb{E}\big[\zeta_tX_t+\int_0^t\!\zeta_u(c_uX_u+\lambda_u(1-\alpha_u)X_u)\,\mathrm{d}u\big]
        \\&\geq\liminf_{t\rightarrow\infty}\mathbb{E}\big[\int_0^t\!\zeta_u(c_uX_u+\lambda_u(1-\alpha_u)X_u)\,\mathrm{d}u\big] 
        \\&=\mathbb{E}\big[\int_0^\infty\!\zeta_u(c_uX_u+\lambda_u(1-\alpha_u)X_u)\,\mathrm{d}u\big]. 
    \end{aligned}
    \end{equation}
    
    Next, we show that the above inequalities are equalities for $X^*$. Reexamining (\ref{eq:Y}) yields
    \begin{equation}\label{eq:Y^*}
    \begin{aligned}
        \mathrm{d}Y^*_t&=\zeta_t\mathrm{d}X^*_t+X^*_t\mathrm{d}\zeta_t+\mathrm{d}[\zeta,X^*]_t+\zeta_t(c^*_tX^*_t+\lambda_t(1-\alpha^*_t)X^*_t)\mathrm{d}t
        \\&=\zeta_tX_t^*\sigma\pi_t^*\mathrm{d}W_t-\zeta_tX_t^*\frac{\mu-r}{\sigma}\mathrm{d}W_t
        \\&=\zeta_tX_t^*\big(\sigma\pi_t^*-\frac{\mu-r}{\sigma}\big)\mathrm{d}W_t.
    \end{aligned}
    \end{equation}
    Reexamining (\ref{eq:X}) and (\ref{eq:zeta}) yields
    \[\zeta_tX_t^*=\varphi_0X_0\exp\Big(-\int_0^t\!c^*_u+\lambda_u(1-\alpha_u^*)\,\mathrm{d}u\Big)\,\mathcal{E}_t\Big(\int_0\sigma\pi^*_u-\frac{\mu-r}{\sigma}\,\mathrm{d}W_u\Big),\]
    in particular, using $\alpha^*,c^*,\pi^*$ deterministic,
    \begin{align}
        \mathbb{E}[\zeta_tX_t^*]&=\varphi_0X_0\exp\Big(-\int_0^t\!c^*_u+\lambda_u(1-\alpha_u^*)\,\mathrm{d}u\Big),\label{eq:1stmoment}
        \\\mathbb{E}[(\zeta_tX_t^*)^2]&\leq(\varphi_0X_0)^2\exp\Big(\int_0^t\!\big(\sigma\pi^*_u-\frac{\mu-r}{\sigma}\big)^2\,\mathrm{d}u\Big).\label{eq:2ndmoment}
    \end{align}
    Combining (\ref{eq:Y^*}) with (\ref{eq:2ndmoment}) and noting that $\mu,r,\sigma,\pi^*$ are constants shows that the local martingale $Y^*$ has an integrable quadratic variation, i.e.\,$Y^*$ is a martingale and the first inequality in (\ref{eq:Y0}) is an equality for $X^*$. Furthermore, let
    \[f_t=\mathrm{e}^{-\beta t}S_t(1+b_t^{\frac{1}{1-\gamma}}\lambda_t)\]
    then using $\alpha^*,c^*$ and the constraint $\int_0^\infty\!f_u\,\mathrm{d}u<\infty\,$ from Proposition \ref{pro:controls} yields
    \begin{equation}\label{eq:-infty}
    \begin{aligned}
        -\int_0^t\!c^*_u+\lambda_u(1-\alpha_u^*)\,\mathrm{d}u&=\int_0^t\!\frac{-f_s}{\int_s^\infty\!f_u\,\mathrm{d}u}\,\mathrm{d}s=\log\int_t^\infty\!f_u\,\mathrm{d}u-\log\int_0^\infty\!f_u\,\mathrm{d}u
        \\&\xrightarrow{t\uparrow\infty}-\infty.
    \end{aligned}
    \end{equation}
    Combining (\ref{eq:1stmoment}) and (\ref{eq:-infty}) shows that the second inequality in (\ref{eq:Y0}) is an equality for $X^*$. 
    
    Overall, (\ref{eq:Y0}) holds with equality for $X^*$, thus
    \[
    \begin{aligned}
        0&=\,Y_0-Y_0
        \\&\geq\mathbb{E}\big[\int_0^\infty\!\zeta_u(c_uX_u+\lambda_u(1-\alpha_u)X_u)\,\mathrm{d}u\big]-\mathbb{E}\big[\int_0^\infty\!\zeta_u(c^*_uX^*_u+\lambda_u(1-\alpha^*_u)X^*_u)\,\mathrm{d}u\big]
        \\&=\mathbb{E}\big[\int_0^\infty\!\mathrm{e}^{-\rho u}\varphi_u(X^*_u)^{\gamma-1}\big(c_uX_u+\lambda_u(1-\alpha_u)X_u-c_u^*X_u^*-\lambda_u(1-\alpha^*_u)X_u\big)\;\mathrm{d}u\big],
    \end{aligned}    
    \]
    which implies that (\ref{eq:noindicators}) is negative. Moreover, (\ref{eq:noindicators}) is bounded because both terms in the above difference are positive and the subtrahend is the finite value $Y_0$. Thus, the sum of the last two summands in (\ref{eq:estimate}) is negative and bounded. Thus, (\ref{eq:estimate}) yields
    \[
    \begin{aligned}
        &\mathbb{E}\big[\int_0^\infty\!\mathrm{e}^{-\rho u}S_u\big(c_uX_u\big)^{\gamma}/\gamma+\mathrm{e}^{-\rho u}S_u\lambda_ub_u\big((1-\alpha_u)X_u\big)^{\gamma}/\gamma\,\mathrm{d}u\big]
        \\&\leq\mathbb{E}\big[\int_0^\infty\!\mathrm{e}^{-\rho u}S_u\big(c^*_uX^*_u\big)^{\gamma}/\gamma+\mathrm{e}^{-\rho u}S_u\lambda_ub_u\big((1-\alpha^*_u)X^*_u\big)^{\gamma}/\gamma\,\mathrm{d}u\big],
    \end{aligned}
    \]
    i.e.\ any control has a lower utility than the candidate controls from Proposition \ref{pro:controls}, which means that the candidate controls are optimal.
\end{proof}

\section{Numerical experiments}\label{section:derivation}

\subsection*{Choosing parameter values}

We choose $\,\mu=10\%$, $\,\sigma=20\%$ and $\,r=3\%$, which roughly coincide with historical values for S\&P 500 and inflation in Western countries according to \citet{Investopedia2024}. We chose that particular index because it is the de facto industry-standard benchmark for equities.

We determine the level of risk aversion $\gamma$ by matching $\pi^*$ from Proposition \ref{pro:controls} with the equity allocation of pension portfolios. We pick a range of parameters because there is no industry consensus about the right equity level in retirement. For example, \citet{PensionTimes2022} states a range of 20\% to 40\%; while governmental institutions like the \citet[Table No IV.2]{SP2003} that adopted drawdown in retirement go as low as 5\%, see \citet{FuentesFullmerGarcia2024} for an overview of the Chilean Pension System and adopting tontines. We include risk aversions exceeding equity levels of over 100\%, even if they violate potential borrowing constraints, to showcase the implications of utility maximisation. We choose $\gamma$ to be $\,0.5, -1, -3, -5, -8\,$ and $-11\,$ corresponding to equity allocations of 350\%, 87.50\%, 43.75\%, 29.17\%, 19.4\% and 14.6\%.

We link $\rho$ to inflation rather than seeing it as a subjective discount rate. We discount our wealth to evaluate the purchasing power with the utility function over time. Note that purchasing power changes with inflation but our preferences $\,x\mapsto x^\gamma/\gamma\,$ are fixed. More precisely, to interpret our results in terms of purchasing power, we need to look at the objective function 
\[\mathbb{E}\big[\int_0^\infty\!S_u\big(\mathrm{e}^{-ru}c_uX_u\big)^{\gamma}/\gamma+S_u\lambda_ub_u\big(\mathrm{e}^{-ru}(1-\alpha_u)X_u\big)^{\gamma}/\gamma\,\mathrm{d}u\big],\]
which coincides with our objective function (\ref{eq:objective}) when we assume
\[\rho=r\gamma.\]

We use the UK period life table for both sexes between ages 65 and 110 from the \citet{HMD-UK-2019} as the basis for our mortality distribution. We fit a continuous distribution to our discrete data because we need a force of mortality $\lambda$. More precisely, we use the Gompertz-Makeham distribution with a force of mortality of
\[\lambda_t= a_1\mathrm{e}^{a_2t}+a_3\quad\text{for years $\,t\geq0\,$ after age 65}\]
and constants $\,a_1,a_2,a_3\geq0$. Minimising the squared distance between our data's annual survival probabilities and the theoretical probabilities of our continuous distribution yields $\,a_1=0.00584$, $\,a_2=0.12150\,$ and $\,a_3=0.0024117$. 

\subsection*{Modelling time preferences}

In the last subsection, we specified all our model parameters except for the function $b$. Here, we use common sense arguments to determine a reasonable form for $b$ and then tweak its shape to avoid obvious strange behaviour.  

We focus on individuals who are considering investing in a tontine. First, we consider an opportunistic individual who recognises that dying is unlikely early in retirement and likely at an advanced age. Thus, to maximise their benefits, they would have a weaker bequest motive early on than at an advanced age, i.e.\ aligning their motive with the likelihood of dying. We interpret such an opportunistic individual as risk-tolerant because they ignore cases when they enter the tontine, die early, and get neither the benefit of many payments nor a bequest. Because we can interpret tolerance to risk in terms of $\gamma$, we can say that $\,\gamma>0\,$ yields preferences like $\lambda_t$ over time $\,t\geq0$.

Now, consider a conservative individual who is afraid of losing money. Losing money for no return early in retirement is a bigger problem for such a person than the inability to bequeath after using/living off their funds in retirement. Since the force of mortality is increasing, we can make a similar statement as for the opportunistic individual, i.e.\ we can say that $\,\gamma<0\,$ yields preferences like $\,1/\lambda_t$ over time $\,t\geq0$.  

Because $\gamma$ changes smoothly from one type of preference to the other, the preference is reasonably given by 
\begin{equation}\label{eq:power-preferences}
    b_t=(\lambda_t)^\gamma\quad\text{for $\,t\geq 0$}.
\end{equation}    
It is worth noting that such preferences fulfil the integrability constraint of Proposition \ref{pro:controls}. To see this, recall that we model the mortality with a Gompertz-Makeham, which has an exponentially increasing force of mortality, i.e.\ $\,u\mapsto\mathrm{e}^{-\beta u}(1+\lambda_u^{\gamma/(1-\gamma)}\lambda_u)\,$ is at most exponentially increasing, but $\,S_u=\exp(-\int_0^u\lambda_s\,\mathrm{d}s)\,$ decreases faster than any exponential decay. In particular, the product of the above two terms decreases faster than any exponential decay and is, therefore, integrable over the positive half-line.

\begin{center}
    \includegraphics[trim=0 10 0 30,width=0.45\linewidth]{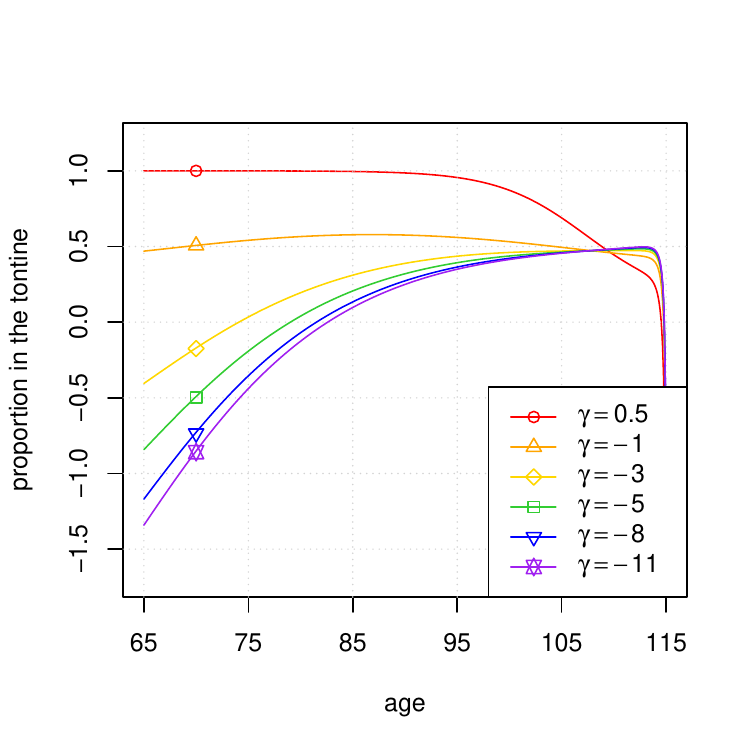} 
    \captionof{figure}{Proportion of wealth allocated to the tontine when the force of mortality directly influences preferences.}
    \label{fig:lambda} 
\end{center}

Figure \ref{fig:lambda} shows the changing proportion of funds in the tontine over time when we use the preferences (\ref{eq:power-preferences}). The curves roughly correspond to our common sense arguments, i.e.\ opportunistic individuals (high $\gamma$) have a decreasing proportion over time while conservative individuals (low $\gamma$) have an increasing one. Focusing on ages 65 to 95, we might see a curve exhibiting an up-and-down movement for some $\gamma$ like $-1$, but we only see a monotone increase for any $\,\gamma\leq-3\,$ corresponding to industry standard levels of equity below 43.75\% for post-retirement funds.

However, there are two issues with Figure \ref{fig:lambda}. First, the initial proportion in the tontine is negative for $\,\gamma\leq-3$. The proportion starts being positive at some age after 75. We could interpret that as tontines being unattractive until age 75, but that would enforce preferences (\ref{eq:power-preferences}) upon an individual willing to invest in a tontine at 65. One way to rectify the situation is to recognise that we only used relative preferences instead of absolute ones in our argumentations. Thus, we can introduce a constant scaling parameter $\,\kappa>0$, i.e.
\begin{equation}\label{eq:power-preferences-scaled}
    b_t=\kappa(\lambda_t)^\gamma\quad\text{for $\,t\geq 0$}.
\end{equation}
That parameter is not free because we know that the individual begins to consider investing in a tontine at age 65, i.e.\ $\kappa$ should be chosen such that $\,\alpha^*_0=0$. As it turns out, this is only possible when $\,\gamma<0$, which coincides with our understanding that $\,\gamma>0\,$ corresponds to an opportunistic individual with an initial allocation of 100\% in the tontine.

\begin{pro} Let $\,\rho=r\gamma\,$ and $\,\lambda_0<\lambda_u\,$ for $\,u>0$. Then, there is $\,\kappa>0\,$ such that $\,b_t=\kappa(\lambda_t)^\gamma\,$ yields $\,\alpha_0^*=0\,$ in Proposition \ref{pro:controls} if and only if $\,\gamma<0$.
\end{pro}
\begin{proof}
    $\alpha_0^*=0\,$ with $\,\kappa>0\,$ is equivalent to $\,\lambda_0^{\frac{\gamma}{1-\gamma}}-\int_0^\infty\mathrm{e}^{-\beta u}S_u\lambda_u^{\frac{1}{1-\gamma}}\,\mathrm{d}u>0$. Let
    \[f(\gamma)=\lambda_0^{\frac{\gamma}{1-\gamma}}-\int_0^\infty\mathrm{e}^{-\beta u}S_u\lambda_u^{\frac{1}{1-\gamma}}\,\mathrm{d}u\]
    be the left-hand side of the inequality. Note that $\beta$ depends on $\gamma$, see Proposition \ref{pro:controls}. Using $\,\rho=r\gamma\,$ and $\,\lambda_0<\lambda_u\,$ for $\,u>0\,$ yields
    \begin{gather*}
        f(0)=1-\int_0^\infty\!S_u\lambda_u\,\mathrm{d}u=1-S(0)=0,
        \\f'(\gamma)< f(\gamma)\frac{\log(\lambda_0)}{(1-\gamma)^2}.
    \end{gather*}
    Applying Gr\"{o}nwall's lemma gives
    \begin{equation*}
        f(\gamma)\;
        \begin{cases}
            <f(0)\exp(\int_0^\gamma\frac{\log(\lambda_0)}{(1-\gamma)^2}\,\mathrm{d}u)=0&\quad\text{for $\,\gamma>0$}
            \\>f(0)\exp(\int_\gamma^0\frac{\log(\lambda_0)}{(1-\gamma)^2}\,\mathrm{d}u)=0&\quad\text{for $\,\gamma<0$}
        \end{cases}.
    \end{equation*}
    Overall,  $\,\lambda_0^{\frac{\gamma}{1-\gamma}}-\int_0^\infty\mathrm{e}^{-\beta u}S_u\lambda_u^{\frac{1}{1-\gamma}}\,\mathrm{d}u>0\,$ if and only if $\,\gamma<0$, which is equivalent to the existence of $\,\kappa>0\,$ such that $\,\alpha_0^*=0$.
\end{proof} 

The plummeting proportion in the tontine around age 115 is the second issue with Figure (\ref{fig:lambda}). The fall is a numerical artefact originating from imposing a limiting age of 115 to deal with unbounded integrals. Extending the limiting age would remove the fall from 115 but reveal other hard-to-explain behaviour at extreme ages. For example, we can already observe two inflexion points in the curve of $\,\gamma=0.5\,$ in Figure (\ref{fig:lambda}). Generally speaking, our objective function (\ref{eq:objective}) tends to give us answers that do not follow common sense at extreme ages. \citet[p.174]{BernhardtDonnelly2019} encountered such a problem already and circumvented it by introducing a limiting age for their setup with a fixed proportion in the tontine, which creates numerical issues in our case. As before, we rectify the situation by adjusting our preferences. More precisely, we remove extreme ages from the bequest motive, e.g.
\begin{equation}\label{eq:power-preferences-trimmed}
    b_t=\Big(\frac{1}{\frac{1}{\lambda_t}-\frac{1}{\lambda_{20}}}\Big)^\gamma\quad\text{for $\,t\in[0,20]$, and $0$ otherwise}.
\end{equation}
We specifically have a person in mind who has a bequest motive early in retirement but realises they need their funds for themselves when they have not died prematurely. This description matches our conservative individual best. Thus, the preferences (\ref{eq:power-preferences-trimmed}) are such that the bequest motive changes continuously from a finite value to zero after a fixed time, here 20 years, when $\,\gamma<0$. To some degree, the preferences are also reasonable for an opportunistic person whose preferences align with an increasing $\lambda$. Their bequest motive is the largest when the likelihood of dying is the largest; hence the bequest will be the largest right before it becomes unavailable. Thus, it makes sense that (\ref{eq:power-preferences-trimmed}) is infinity when a bequest becomes unavailable for $\,\gamma>0$. However, even though the preferences (\ref{eq:power-preferences-trimmed}) yield well-defined controls, they do not fulfil the integrability constraint of Proposition \ref{pro:controls} for $\,\gamma>0$.   

\begin{center}
    \includegraphics[trim=0 10 0 30,width=0.45\linewidth]{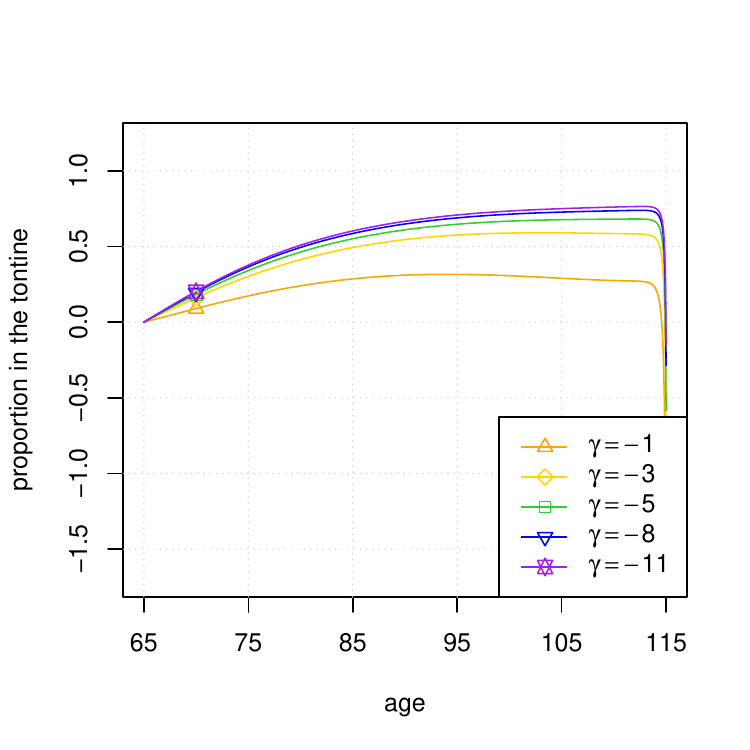}
    \includegraphics[trim=0 10 0 30,width=0.45\linewidth]{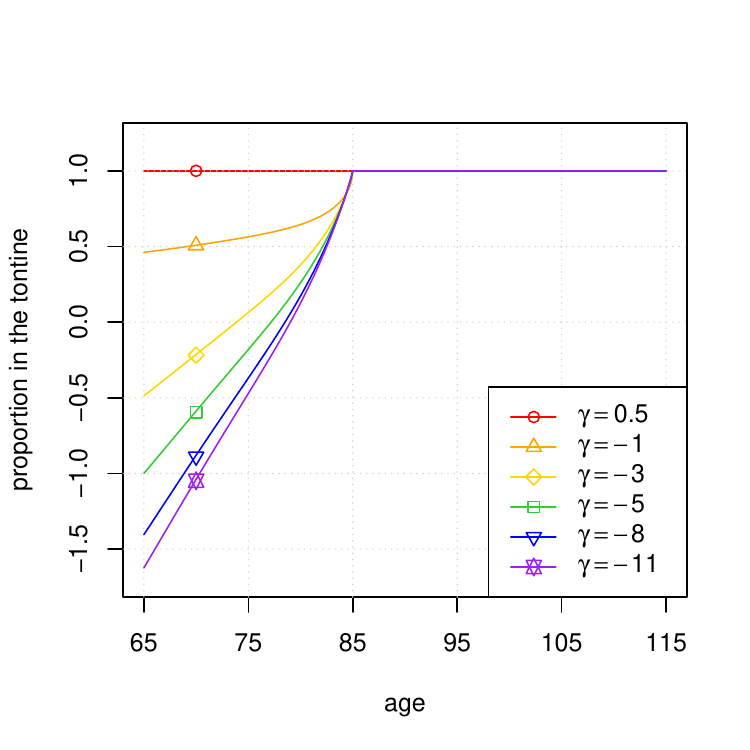}
    \captionof{figure}{Proportion of wealth allocated to the tontine when initial allocation has to be 0 at initial age 65 (left) and when bequest motive vanishes continuously after 20 years (right).  
    }
    \label{fig:lambda-scale-or-trim} 
\end{center}

Figure \ref{fig:lambda-scale-or-trim} shows the changing fund allocation in the tontine over time when we use the preferences (\ref{eq:power-preferences-scaled}) on the left and the preferences (\ref{eq:power-preferences-trimmed}) on the right. We removed $\,\gamma=0.5\,$ from the left picture because there is no scale $\,\kappa>0\,$ yielding $\,\alpha^*_0=0\,$ in this case. We see on the left that requiring $\,\alpha^*_0=0\,$ yields proportions that are all positive. In addition, note that the curves are fairly close together when $\,\gamma\leq-3$, indicating some robustness. And, we see on the right that excluding advanced ages from bequest motives yields consistently monotone increasing proportions. Indeed, both adjustments (\ref{eq:power-preferences-scaled}) and (\ref{eq:power-preferences-trimmed}) remove the strange behaviour observed in Figure \ref{fig:lambda}.

\begin{center}
    \includegraphics[trim=0 10 0 30,width=0.45\linewidth]{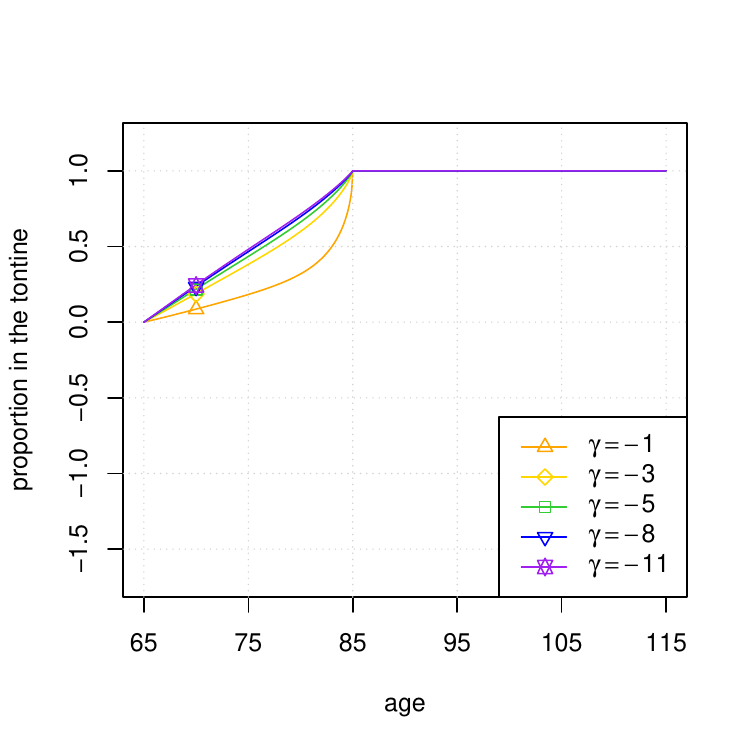} 
    \captionof{figure}{Proportion of wealth allocated to the tontine when initial allocation has to be 0 at initial age 65 in addition to a vanishing bequest motive after 20 years.}
    \label{fig:lambda-scale-and-trim} 
\end{center}

Figure \ref{fig:lambda-scale-and-trim} shows the changing fund allocation in the tontine when we combine the two previous adjustments, i.e.\ it depicts the proportion based on our most reasonable bequest preferences, which are 
\begin{equation}\label{eq:power-preferences-scaled-and-trimmed}
    b_t=\kappa\Big(\frac{1}{\frac{1}{\lambda_t}-\frac{1}{\lambda_{20}}}\Big)^\gamma\quad\text{for $\,t\in[0,20]$, and $0$ otherwise},
\end{equation}
with $\,\kappa>0\,$ such that the initial allocation in the tontine is 0, i.e.\;$\alpha_0^*=0\,$ and $\,\gamma<0$. In addition to seeing no issues like in Figure \ref{fig:lambda},  we see an almost linear change in the tontine allocation for $\,\gamma\leq-3\,$ corresponding to industry standard levels of equity. We would expect a linear adjustment in practice compared to an alternative complicated one. Even though we cannot find examples because modern tontines do not exist in the industry yet to our knowledge, we can find pension products with bequest benefits that linearly adjust bequest benefits over time like \citet[pp.7-8]{Challenger2022}; note that a linear adjustment in the tontine allocation yields a complementary linear adjustment in the bequest benefit according to our objective (\ref{eq:objective}).   

\subsection*{Sanity check}

Previously, we focused exclusively on the fund allocation in the tontine. We tweaked the unknown bequest preference $b$ until the proportion $\alpha^*$ followed some rules. However, that might have resulted in nonsensical changes in the other controls, especially $c^*$ (the control $\pi^*$ is independent of $b$). We look at discounted average income rate $\,t\mapsto\mathrm{e}^{-rt}c^*_tX^*_t\,$ because we try to optimise purchasing power and it coincides with annual income up to some discretisation error. For simplicity, we take the expectation 
\[t\mapsto\mathbb{E}[\mathrm{e}^{-rt}c^*_tX^*_t]\quad\text{for $\,t\geq0$}\]
because volatility would obscure the picture. 

\begin{center}
    \includegraphics[trim=0 10 0 30,width=0.45\linewidth]{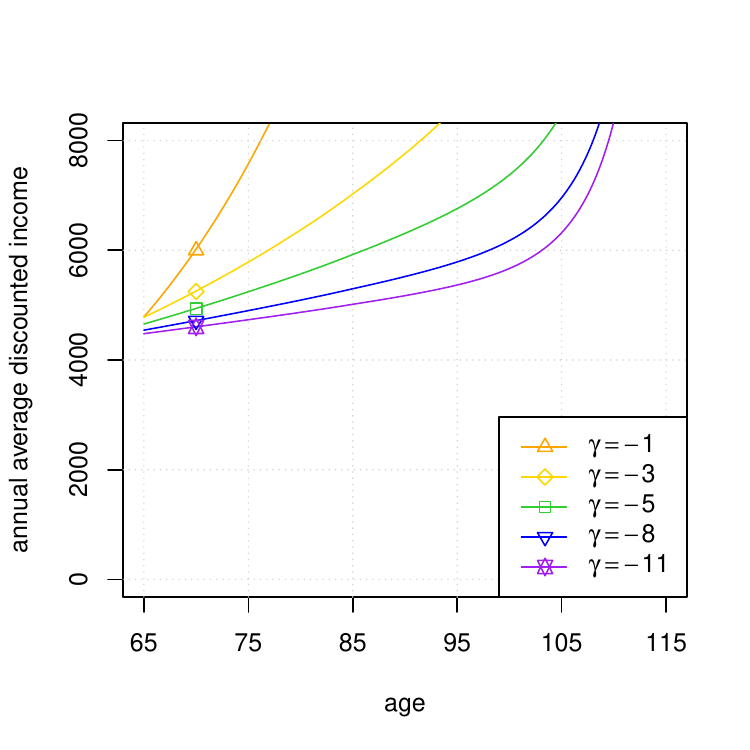}
    \captionof{figure}{Annual income of a person aged 65 with an initial investment of \pounds100k with preferences that yield an initial allocation of 0 in the tontine and continuously shift to 1 after 20 years.}
    \label{fig:income-lambda-scaled-and-trimmed} 
\end{center}

Figure \ref{fig:income-lambda-scaled-and-trimmed} shows the annual income based on the preferences (\ref{eq:power-preferences-scaled-and-trimmed}). The most important detail is the initial value between \pounds4480 to \pounds4786, which fits the annual income between \pounds4540 to \pounds4756 of available annuities escalating at retail price index on the market according to \citet{Retirementline2024}. In particular, our chosen preferences have kept consumption at reasonable levels. 

It is (actually) surprising that the income levels between our tontines and annuities match up so tightly. The numbers line up perfectly to balance out the bequest benefit of the tontines with the fees for annuities yielding equal income levels. That the preferences (\ref{eq:power-preferences-scaled-and-trimmed}) remove the bequest benefit after 20 years is essential here. If we preferred bequest during the whole lifetime like with (\ref{eq:power-preferences}), we would see similar shaped curves as in Figure \ref{fig:income-lambda-scaled-and-trimmed} but with a significantly lower initial consumption between \pounds2872 to \pounds3833 for $\,\gamma\leq-3$. Given that different retirement products compete on the market and people avoid trade-offs, we imagine that preferences tend to match the income levels of available retirement products. More precisely, we imagine that the time frame of bequest benefits, 20 years in (\ref{eq:power-preferences-scaled-and-trimmed}), would be used in practice to match the income of other retirement products. 

It is worth noting that $\,\mu>r\,$ is the reason for the monotone increase in consumption in Figure \ref{fig:income-lambda-scaled-and-trimmed}. We like to see constant consumption, but only specific combinations of $\gamma$ and $\rho$ yield that for all $\,\mu-r$. Those combinations are $\,\gamma=2/3\,$ or $\,\gamma=-\infty\,$ when $\,\rho=r\gamma$, i.e.\ 525\% or 0\% in equity respectively, which makes no sense. Most notably, we cannot find preferences $b$ that yield constant annual average discounted consumption because the precise condition depends only on externally given parameters and $\,\mu-r\neq0$. More precisely:

\begin{pro}
    In the context of Proposition \ref{pro:controls},
    \[\mathbb{E}[\mathrm{e}^{-rt}c^*_tX^*_t]=\frac{\mathrm{e}^{-\beta t}\mathrm{e}^{(\mu-r)\pi^*t}}{\int_0^\infty\mathrm{e}^{-\beta u}S_u(1+b_u^{\frac{1}{1-\gamma}}\lambda_u)\,\mathrm{d}u},\]
    which is constant if and only if $\,(\mu-r)\pi^*-\beta=0$.
\end{pro}
\begin{proof}
    Let
    \[f_t=\mathrm{e}^{-\beta t}S_t(1+b_t^{\frac{1}{1-\gamma}}\lambda_t)\quad\text{for $\,t\geq0$}.\]
    Similar to (\ref{eq:-infty}), combining (\ref{eq:X}) with the definitions of $\pi^*,c^*,\alpha^*$ from Proposition \ref{pro:controls} and noting that $\pi^*$ is constant, i.e.\ $\,\mathcal{E}_t(\int_0\sigma\pi_u^*\,\mathrm{d}W_u)\,$ is a martingale, yields
    \begin{align*}
        \mathbb{E}[\mathrm{e}^{-rt}c^*_tX^*_t]&=\frac{\mathrm{e}^{-\beta t}\mathrm{e}^{(\mu-r)\pi^*t}\exp(\int_0^t\frac{-f_u}{\int_u^\infty\!f_s\,\mathrm{d}s}\,\mathrm{d}u)}{\int_t^\infty\!f_u\,\mathrm{d}u}
        \\&=\frac{\mathrm{e}^{-\beta t}\mathrm{e}^{(\mu-r)\pi^*t}\exp\big(\log(\int_t^\infty\!f_u\,\mathrm{d}u)-\log(\int_0^\infty\!f_u\,\mathrm{d}u)\big)}{\int_t^\infty\!f_u\,\mathrm{d}u}=\frac{\mathrm{e}^{-\beta t}\mathrm{e}^{(\mu-r)\pi^*t}}{\int_0^\infty\!f_u\,\mathrm{d}u}.
    \end{align*}
\end{proof} 

\section{Conclusion}

We have studied the optimal allocation of funds in a modern tontine with bequest. We verified that the candidate allocation from a corresponding HJB approach is optimal, a mathematical issue raised by \citet[p.264]{Dagpunar2021} and left for future research. We have discussed why some aspects of the allocation violate common sense and have suggested solutions.  

We have used duality methods from stochastic optimal control theory to verify the candidate controls in our utility maximisation problem. Those tools are particularly well-suited when dealing with the power utility. We hope the paper raises awareness of those tools in the actuarial community. For example, we can verify the candidate controls in related publications like \citet{HeLiangRen2024} and \citet{NgNguyen2024} instead of assuming hard-to-check integrability constraints.

It is clear that investing in a tontine and then losing everything because of sudden death is an unacceptable circumstance, similar to why bequest is a reason for the annuity puzzle. But, nobody sharing their funds means a low income for everyone. As a result, we argued that the initial investment in the tontine should increase from 0 to 1 in a fixed time interval. We have found preferences that model those two constraints in a power utility set-up. Those preferences yield almost linear allocation in the tontine for sensible ranges of risk aversion, i.e.\ one of the simplest ways to allocate funds to the tontine is likely the most preferred and practical one.

\emergencystretch1em 
\printbibliography 

\end{document}